%% file: main.tex
\documentclass[10 pt, conference, letterpaper]{IEEEtran}

\IEEEoverridecommandlockouts                              
\usepackage{graphicx,epstopdf} 
\usepackage{amsmath,amssymb,amsfonts,amsthm}
\usepackage{multirow}
\usepackage{environ}         % provides \BODY
\usepackage{etoolbox}        % provides \ifdimcomp
\usepackage{textcomp}
\usepackage{algorithm,algpseudocode}
\usepackage{xcolor}
\usepackage{pgfplots}
\usepackage{filecontents}
\pgfplotsset{compat=1.16}

\usepackage[colorlinks=true, bookmarks=true]{hyperref}
\AtBeginDocument{%
  \hypersetup{
    citecolor=red,
    linkcolor=blue,   
    urlcolor=blue}}

\usepackage{tikz}
\usepackage{mathtools}
\usepackage{setspace}
\usepackage{array}
\usepackage{calc}
\usepackage{makecell}
\usepackage{tabularx}
\usepackage{booktabs}
\usepackage{makecell}
\usepackage{pgfplots,pgfplotstable}
\pgfplotsset{compat=1.16}
\usepgfplotslibrary{statistics}
\usepackage{scrextend}
\usepackage{siunitx}
\usepackage{caption}
\usepackage{subcaption}
\usepackage{cite}

\allowdisplaybreaks

\makeatletter
\newcommand*{\rom}[1]{\expandafter\@slowromancap\romannumeral #1@}
\makeatother

\newtheorem{lemma}{Lemma}
\newtheorem{theorem}{Theorem}
\newtheorem{definition}{Definition}

\newtheorem{remark}{Remark}

\newcommand{\ac}{algebraic connectivity}
\newcommand{\Ac}{Algebraic connectivity}

\title{\LARGE \bf
Spectral Graph Theoretic Methods for Enhancing \\ Network Robustness in Robot Localization
}

\author{Neelkamal Somisetty$^{1}$, Harsha Nagarajan$^{2,*}$, Swaroop Darbha$^{1}$
\thanks{*The authors gratefully acknowledge funding from Triad National Security LLC under the grant from the DOE NNSA (award no. 89233218CNA000001), titled ``Modeling and Discrete Optimization Algorithms for Robust Complex Networks'' and U.S. DOE's Laboratory Directed Research \& Development program under the project ``20230091ER: Learning to Accelerate Global Solutions for Non-convex Optimization''. \vspace{0.1cm}}
\thanks{$^{1}$Department of Mechanical Engineering, Texas A \& M University, College Station, TX, USA. 
Email: \{\href{mailto:neelkamal.sept18@tamu.edu}{neelkamal.sept18}, \href{mailto:dswaroop@tamu.edu}{dswaroop}\}@tamu.edu
}
\thanks{$^{2,*}$(Corresponding author) Applied Mathematics and Plasma Physics (T-5), Los Alamos National Laboratory, Los Alamos, NM, USA. Email:  \href{mailto:harsha@lanl.gov}{harsha@lanl.gov}}}

\begin{document}

\maketitle
\thispagestyle{empty}
\pagestyle{empty}

\begin{abstract}

This paper addresses the optimization of edge-weighted networks by maximizing {\ac} to enhance network robustness. Motivated by the need for precise robot position estimation in cooperative localization and pose-graph sparsification in Simultaneous Localization and Mapping (SLAM), the {\ac} maximization problem is formulated as a Mixed Integer Semi-Definite Program (MISDP), which is NP-hard. Leveraging spectral graph theoretic methods, specifically Cheeger's inequality, this work introduces novel ``\textit{Cheeger cuts}'' to strengthen and efficiently solve medium-scale MISDPs. Further, a new Mixed Integer Linear Program (MILP) is developed for efficiently computing Cheeger cuts, implemented within an outer-approximation algorithm for solving the MISDP. A greedy $k$-opt heuristic is also presented, producing high-quality solutions that serve as valid lower bounds for Cheeger cuts. Comprehensive numerical analyses demonstrate the efficacy of strengthened cuts via substantial improvements in run times on synthetic and realistic robot localization datasets. 

\end{abstract}

\section{INTRODUCTION}
\label{Sec:intro}

Localization is pivotal for the autonomy of mobile robots, enabling precise determination of sensor positions in various critical domains, including environmental monitoring, surveillance, mapping in GPS-challenged environments, object tracking, and geographic routing \cite{boukerche2007localization,patwari2005locating}. 
A prominent method in this discipline is Simultaneous Localization and Mapping (SLAM), wherein robots use exteroceptive sensors (e.g., cameras, lasers, sonars) along with motion data from interoceptive sensors (e.g., IMUs, encoders) to concurrently construct a map and determine their location \cite{cadena2016past}. In applications involving multiple robots, cooperative localization becomes essential, relying on the collaborative efforts of robots to share relative position and motion data for joint state estimation \cite{roumeliotis2002distributed, lajoie2023swarm}. This necessity becomes more pronounced in dynamic networks, where consistent connectivity is not guaranteed, requiring each robot to rely on odometry and shared relative measurements for localization \cite{leung2009decentralized}.  

In such networks, ensuring robust connectivity among communication links is crucial for effectively assessing the states of robots. Typically, robot networks operate on the premise that each robot can seamlessly exchange information to synchronize its actions with others. However, the inherent network connectivity may be compromised due to measurement errors, sensor failures, communication delays, and disturbances. {\Ac}, the second smallest eigenvalue of the Laplacian matrix, serves as a key metric for robustness and has garnered attention from both graph theoretical and engineering perspectives, underscoring its importance in enhancing network resilience \cite{jamakovic2007relationship}. {\Ac} is used as a maximization objective to sparsify pose-graphs in SLAM, addressing the escalating memory demands for observation storage and the computational burden of state estimation algorithms in long-term navigation \cite{doherty2022spectral}.

This paper addresses a specific instance of the robust network synthesis problem, \textit{which remains unresolved}, and is significant in localization applications. The challenge involves finding a subgraph with up to $q$ edges from a weighted graph to achieve the maximum possible {\ac}, known as the {\ac} maximization problem. This problem's NP-hard classification underscores its complexity \cite{mosk2008maximum}. In prior research \cite{nagarajan2012algorithms, nagarajan2015maximizing, somisetty2024optimal}, we employed a mixed-integer semi-definite program (MISDP) framework using outer-approximation (OA) and cutting planes (cuts) for obtaining optimal solutions. However, this approach has computational limitations when identifying optimal spanning trees in complete graphs exceeding ten nodes, necessitating strengthening of cuts. Leveraging spectral graph theory, this paper develops a \textit{novel} valid cut by utilizing the relationship between algebraic connectivity and the Cheeger constant \cite{cheeger1970lower} of a weighted graph, thereby enhancing the efficiency of solving the MISDP. The contributions of this paper are: 
\begin{list}{\labelitemi}{\leftmargin=0.8em \itemindent=-0.35em}
    \item Introducing valid cuts based on the Cheeger constant of a graph to strengthen MISDP bounds, leading to better run times for larger networks.
    \item  Developing a novel mixed-integer linear program (MILP) for efficient computation of the Cheeger constant in undirected weighted graphs. 
    \item Developing a greedy $k$-opt heuristic to quickly generate high-quality solutions for the pose-graph SLAM sparsification problem, involving benchmark datasets such as \texttt{Intel} with 1728 nodes and 400 loop-closure edges.
\end{list}

The paper is organized as follows: Section \ref{Sec:lit} reviews recent findings on maximizing {\ac} in robot localization. Section \ref{sec:formulation} presents the mathematical framework as an MISDP with connectivity constraints. Section \ref{sec:cheeger} introduces a novel Cheeger constant-based cut formulation and an algorithm using these cuts. Section \ref{sec:heuristic} discusses a greedy heuristic for large-scale networks. Sections \ref{sec:comp} and \ref{sec:conc} present numerical results and conclusions, respectively.

\section{RELATED WORK}
\label{Sec:lit}
\subsection{Importance of the {\ac} in localization}

The significance of {\ac} spans various domains, particularly in localization and perception, where it is crucial for estimation accuracy in cooperative localization \cite{sharma2014observability, nagarajan2015maximizing}, rotation averaging \cite{boumal2014cramer}, linear SLAM \cite{ khosoussi2019reliable}, and pose-graph SLAM \cite{chen2021cramer}. In cooperative vehicle localization with noisy measurements, the state estimation error of the covariance matrix is inversely related to the norm of the Laplacian matrix, suggesting that higher {\ac} leads to reduced estimation errors \cite{somisetty2024optimal}. Similarly, \cite{boumal2014cramer} observed that {\ac} inversely bounds the Cramér-Rao lower bound on the expected mean squared error for rotation averaging. Recent studies have further emphasized its role in minimizing worst-case errors of estimators in measurement graphs for pose-graph SLAM and rotation averaging, linking higher {\ac} to statistically reduced error \cite{doherty2022performance}.

In applications like cooperative vehicle localization with noisy measurements \cite{somisetty2024optimal} and pose-graph sparsification for SLAM \cite{doherty2022spectral}, the central problem involves selecting $q$ edges from a predetermined set to form a subgraph with the highest possible {\ac}. Here, the edge weights, derived from relative measurements or poses, are non-negative and reflect the reliability of the corresponding information.

\subsection{Maximizing the {\ac}}

The problem of maximizing the {\ac} subject to cardinality constraints has been explored in various applications. In \cite{ghosh_growing_2006}, authors considered a binary relaxation of the {\ac} maximization problem, resulting in weak upper bounds. For sparsifying pose-graphs in SLAM, the Frank-Wolfe method was implemented on the continuous semi-definite program (SDP) to identify networks with higher {\ac} \cite{doherty2022spectral, nam2023spectral}. 
However, these methods primarily yield heuristic solutions without optimality guarantees. The literature also includes various heuristic strategies \cite{trimble2019connectivity,nagarajan2014heuristics, son2010building} aimed at generating sub-optimal solutions for  {\ac} maximization. Despite these advancements, there exists a gap in methodologies developed to obtain optimal solutions for maximizing {\ac} in weighted networks.

\section{MATHEMATICAL FORMULATION}
\label{sec:formulation}
This section outlines the mathematical framework for {\ac} maximization, focusing on pose-graph SLAM and cooperative localization applications.

\textbf{Notation.} In subsequent sections, scalars (vector/matrix entries) are denoted by lower and upper case letters, while vectors and matrices are represented in bold font for lower and upper cases, respectively. The tensor product of two vectors $\mathbf{v_1, v_2}$ in the same vector space is expressed as $\mathbf{v_1 \otimes v_2}$, and their scalar (dot) product is denoted by $\mathbf{v_1 \cdot v_2}$. For any two square symmetric matrices $\mathbf{A}$ and $\mathbf{B}$, the notation $\mathbf{A} \succeq \mathbf{B}$ signifies that $\mathbf{A-B}$ is positive semi-definite (PSD), i.e., $\mathbf{A-B} \succeq 0$. Here, $\mathbf{e_i}$ represents the $i^{th}$ column of the identity matrix $\mathbf{I}_n$, where $n \times n$ denotes the matrix size. Additionally, $\mathbf{1}$ denotes an $n$-dimensional vector with all components equal to $1$. For any non-empty set $S$, $|S|$ indicates its cardinality.

Let $ (V, E, \mathbf{w}) $ denote a weighted graph or network with at most one edge between any pair of nodes and no self-loops. The number of nodes is $ n = |V| $. Each edge $\{i,j\} \in E$ has a positive weight $ w_{ij} > 0 $ and a binary decision variable $ x_{ij} \in \{0, 1\}$, where $ x_{ij} = 1 $ indicates the edge is selected, and $ x_{ij} = 0 $ indicates its exclusion. The matrix $ \mathbf{x} $ consists of the decision variables $ x_{ij} $. For a subset $ S \subset V $, $ \delta(S) $ represents the set of edges in the cut-set of $ S $, defined as $ \delta(S) = \{\{i,j\} \in E \mid i \in S, \ j \in V \setminus S \} $, where $ V \setminus S $ is the complement of $ S $ in the vertex set $ V $.

The Laplacian matrix of the weighted graph is given by:
$$\mathbf{L(x)} = \sum_{\{i,j\} \in E} x_{ij} w_{ij}(\mathbf{e}_i - \mathbf{e}_j) \otimes (\mathbf{e}_i - \mathbf{e}_j).$$

Note that $\mathbf{L(x)}$ is a symmetric PSD matrix for a given network $\mathbf{x}$. Let $\lambda_1 ( = 0) \leqslant \lambda_2 \leqslant \lambda_3 \leqslant \ldots \leqslant \lambda_n$ be the eigenvalues of $\mathbf{L(x)}$ and $\mathbf{v}_1, \mathbf{v}_2, \ldots, \mathbf{v}_n$ be the respective eigenvectors, where $\lambda_2$ and $\mathbf{v}_2$ are known as the {\ac} and Fiedler vector, respectively.

The basic problem ($\mathcal{BP}$) is formulated as follows:
\begin{equation}
		\begin{array}{lll}
	  		\text{($\mathcal{BP}$)}  &\gamma^* = &\max  \lambda_2(\mathbf{L(x)}), \\
			&\text{s.t.} & \sum_{\{i,j\} \in E} x_{ij} \leqslant q, \\
			& & \mathbf{x} \in \{0, 1\}^{|E|},
		\end{array}
		\label{eq:BP}
\end{equation}
where $q$ is a prespecified positive integer serving as the upper limit on the number of edges selected. Given this program's non-linear and binary nature, it is crucial to reformulate it into a more computationally feasible structure. The subsequent discussion presents an equivalent MISDP formulation for $\mathcal{BP}$.

\subsection{Mixed Integer Semi-Definite Program}
\label{Subsec:F1}

In the following discussion, we explore two variations of $\mathcal{BP}$. The first is a broad-based approach that augments a base graph to achieve the desired number of edges. This is particularly relevant for sparsifying pose-graphs in SLAM, where the base graph is often a path graph, representing relative pose measurements from encoders. Augmentation involves adding loop closure edges, symbolizing relative pose measurements from exteroceptive sensors.

The second variant considers weighted spanning trees as feasible solutions, with the optimal solution ($\mathbf{x}^*$) being a spanning tree with maximum {\ac}. Spanning trees, representing minimally connected networks, are pertinent to cooperative robot localization. Unlike the first variant, this formulation aims to construct a graph with $n-1$ edges from an empty base graph. However, the algorithms proposed can be adapted to various network configurations.

\textcolor{black}{To simplify the remaining sections, we introduce a lifted PSD matrix, \( \mathbf{W} = \mathbf{L(x)} - \gamma (\mathbf{I}_n - \mathbf{e}_0 \otimes \mathbf{e}_0 ) \), where \( \mathbf{e}_0 = \frac{1}{\sqrt{n}}\mathbf{1} \), such that \( \lVert \mathbf{e}_0 \rVert_2 = 1 \). We then reformulate the \( \mathcal{BP} \) into the following MISDP for the first variant:} 
\vspace{1.6cm}
\begin{subequations}
\begin{align}
    \gamma^* = & \ \max  \ \gamma, \label{eq:gamma_obj}\\  \text{s.t.} \ &  \mathbf{W^a} + \mathbf{W^b}  \succeq 0, \label{eq:W_psd_}\\
		& W_{ii}^a = \sum_{\{i,j\} \in E^a} w_{ij}x_{ij} - \frac{\gamma(n-1)}{n}, \ \forall  i \in V, \label{eq:W_ii_a} \\
		& W_{ij}^a = W_{ji}^a = -w_{ij}x_{ij} + \frac{\gamma}{n}, \ \forall  \{i,j\} \in E^a, \label{eq:W_ij_a} \\
        & W_{ii}^b = \sum_{\{i,j\} \in E^b} w_{ij}, \ \forall  i \in V, \label{eq:W_ii_b} \\
        & W_{ij}^b = W_{ji}^b = -w_{ij}, \ \forall  \{i,j\} \in E^b, \label{eq:W_ij_b} \\
        &\sum_{\{i,j\} \in E^a} x_{ij} \leqslant q, \quad \mathbf{x} \in \{0, 1\}^{|E^a|}.
\end{align}
\label{eq:F_1}
\end{subequations}
In the above formulation, the objective \eqref{eq:gamma_obj} through constraint \eqref{eq:W_ij_b} reflects the {\ac} of the network. The validity of this formulation can be found in \cite{nagarajan2012algorithms}. Here, \(E^a\) denotes the set of edges available for augmentation, while \(E^b\) corresponds to the edges already existing within the base graph. For the variant focusing on spanning trees, we set \(q = n-1\) in \(\mathcal{BP}\), indicating that the ideal network configuration is a sub-graph with \(n-1\) edges, making each network a spanning tree. Consequently, \({\mathcal T}\) represents the exhaustive set of \(n^{n-2}\) possible spanning trees. Since this variant lacks an initial base graph, \(E^a\) is equivalent to \(E\), encompassing all potential edges. The resulting MISDP formulation is as follows:
\begin{subequations}
\begin{align}
    \text{(${\mathcal F}_1$)} \quad \gamma^* = & \ \max  \ \gamma, \\  \text{s.t.} \ &  \mathbf{W}  \succeq 0, \label{eq:W_psd}\\
    & \  \text{Constraints} \  \eqref{eq:W_ii_a}, \eqref{eq:W_ij_a}, \\
		& \mathbf{x} \in \mathcal T \label{eq:xinT}.
\end{align}
\label{eq:F_2}
\end{subequations}
Solving the network design problem (\eqref{eq:F_1} or \eqref{eq:F_2}) using general-purpose MISDP solvers is extremely challenging due to its computational complexity \cite{mosk2008maximum}. In our previous study, we introduced a simplified approach utilizing a cutting plane-based outer-approximation procedure. For a comprehensive explanation of this approach, please refer to \cite{somisetty2024optimal}.

\section{CHEEGER CONSTANT-BASED VALID CUTS}
\label{sec:cheeger}
This section present a novel MILP formulation to compute the Cheeger constant or the Isoperimetric number of a given weighted graph. Subsequently, an algorithm is introduced to find a valid Cheeger cut informed by the well-known Cheeger's inequality generalized for weighted graphs\cite{cheeger1970lower}. These Cheeger cuts, in conjunction with eigenvector cuts, are utilized for the accelerated iterative refinement of an outer-approximation (OA) of the MISDP in ${\mathcal F}_1$ in \eqref{eq:F_2}.

\subsection{MILP formulation for Cheeger constant evaluation}
\label{subsec:milp_cheeger}
\begin{definition}
\label{def1}
Let $G = (V, E, \mathbf{w})$ be an undirected weighted graph. The Cheeger ratio for a vertex set $S$ is 
\begin{align}
    \phi(S) = \frac{\sum_{\{i,j\} \in \delta(S)} w_{ij}}{|S|}.
    \label{eq:cheeger_ratio}
\end{align}
\end{definition}
\vspace{0.2cm}
\begin{definition}
\label{def2}   
For any undirected weighted graph $G = (V, E, \mathbf{w})$, the Cheeger constant, also known as the Isoperimetric number or the edge expansion, represents the sparsest normalized cut-set of the graph and is given by:
\begin{align}
    \phi(G) = \min_{S: \ 1\leqslant |S|\leqslant \frac{|V|}{2}} \phi(S).
    \label{eq:cheeger}
\end{align}
\end{definition}
\vspace{1em}
Note that the Cheeger constant for a weighted graph can also be defined in terms of its \textit{conductance} based on the volume of $S$ \cite{arora2009expander}. However, we adhere to the definition given in \eqref{eq:cheeger}.

Computing $\phi(G)$ in \eqref{eq:cheeger} or the associated cut-set is NP-hard. We propose the following MI quadratically-constrained problem (MIQCP) to evaluate $\phi(G)$ (see \cite{trevisan2013lecture}). To represent a cut $\delta(S)$ that partitions the graph, we introduce a binary vector $\mathbf{z} \in \{0,1\}^n$, where $z_i = 1$ if and only if $i \in S$; otherwise $i \in V\setminus S$. The MIQCP can be written as follows: 
\begin{subequations}
    \begin{align}
    \phi(G) = & \ \min  \ \phi, \\
    \text{s.t.} \ & \phi (\mathbf{1}\cdot \mathbf{z}) \geqslant \sum_{\{i,j\} \in E} w_{ij} (z_i - z_j)^2, \label{eq:r-soc} \\ 
    & 1 \leqslant \mathbf{1}\cdot \mathbf{z} \leqslant \left\lfloor \frac{n}{2} \right\rfloor, \quad \mathbf{z} \in  \{0,1\}^n.
    \end{align}
    \label{eq:miqcp_cheeger}
\end{subequations}
Note that \eqref{eq:miqcp_cheeger} is a mixed-integer convex program with a convex rotated second-order conic constraint in \eqref{eq:r-soc}. This MIQCP can be solved using state-of-the-art solvers like Gurobi \cite{gurobi}; however, such solvers often do not converge or require large run times even on small graphs. Moreover, since $\phi(G)$ is evaluated repeatedly within the branch-and-cut framework of Algorithm \ref{alg:cheeger}, we propose an exact MILP reformulation that achieves faster run times on larger graphs: 
\begin{subequations}
    \begin{align}
    \phi(G) = & \ \min  \ \phi, \\
    \text{s.t.} \ \ & \mathbf{1}\cdot \widehat{\boldsymbol{\phi} \mathbf{z}} \geqslant \sum_{\{i,j\} \in E} w_{ij} (z_i + z_j - 2\widehat{Z}_{ij}), \label{eq:rsoc_linear} \\ 
    & \widehat{Z}_{ij} \leqslant \min\{z_i, z_j\}, \ \forall \{i,j\} \in E, \label{eq:Z_lin_1}\\ 
    & \widehat{Z}_{ij} \geqslant \textcolor{black}{\max\{0, z_i+z_j-1\}}, \ \forall \{i,j\} \in E, \label{eq:Z_lin_2}\\ 
    & \widehat{\phi z}_i \geqslant \max \{0, \ (\phi + \overline{\phi} \cdot z_i - \overline{\phi}) \}, \ \forall i \in V \label{eq:phi_z_1}\\ 
    & \widehat{\phi z}_i \leqslant \min \{\phi, \ \overline{\phi} \cdot z_i \}, \ \forall i \in V \label{eq:phi_z_2} \\ 
    & 1 \leqslant \mathbf{1}\cdot \mathbf{z} \leqslant \left\lfloor \frac{n}{2} \right\rfloor, \quad \mathbf{z} \in  \{0,1\}^n. \label{eq:sum_z}
    \end{align}
    \label{eq:milp_cheeger}
\end{subequations}
where $\overline{\phi}$ is any valid upper bound on the Cheeger constant $\phi$ of the graph. Constraint \eqref{eq:rsoc_linear} is derived by introducing auxiliary variables $\widehat{\phi z}_i$ and $\widehat{Z}_{ij}$, representing the products $\phi z_i$ and $z_i z_j$, respectively. Constraints \eqref{eq:Z_lin_1} and \eqref{eq:Z_lin_2} linearize the binary product $z_i z_j$ over the set of binary points with its convex hull, while constraints \eqref{eq:phi_z_1} and \eqref{eq:phi_z_2} represent the McCormick linearization of the product of the continuous variable $\phi$ and the binary variable $z_i$. These linearizations are exact, ensuring that the optimal solution to these constraints is also optimal for formulation \eqref{eq:miqcp_cheeger}. Finally, constraint \eqref{eq:sum_z} ensures that the graph is partitioned into $S$ and $V \setminus S$, with $S$ containing at least one node but not exceeding half the total number of nodes in the graph.
\begin{figure}
    \centering
    \includegraphics[width=\linewidth]{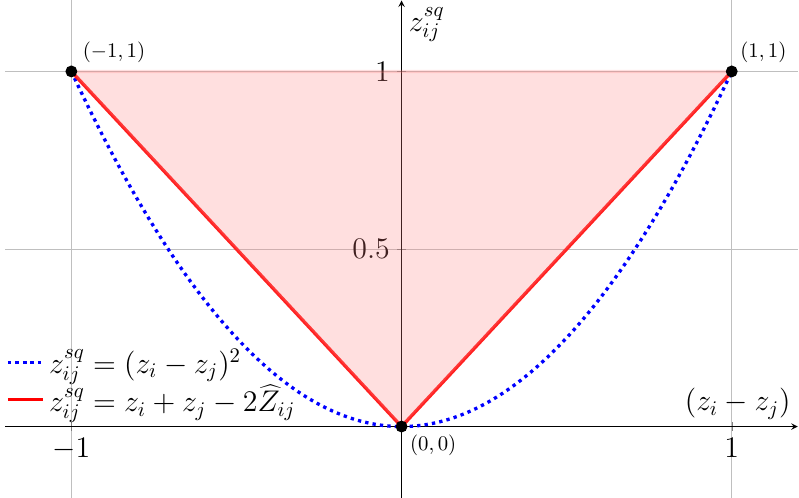}
    \caption{Convex hull as the feasible region (in red) for every line $\{i,j\} \in E$.}
    \label{fig:convex_hull}
\end{figure}

% \vspace{0.1cm}
\begin{remark}
\label{rem:milp_region}
For each line \(\{i,j\}\in E\), with \(z_i,z_j \in \{0,1\}\), observe that the feasible region of \((z_i-z_j)^2\) in it's epigraph form in constraint \eqref{eq:r-soc} is the region above the dotted blue line in Fig. \ref{fig:convex_hull}\textcolor{black}{, labeled as \(z^{sq}_{ij}\)}. In contrast, the feasible region induced by $(z_i+z_j-2\widehat{Z}_{ij})$'s epigraph in constraint \eqref{eq:rsoc_linear} forms the convex hull within the shaded red region, \textcolor{black}{also labeled as \(z^{sq}_{ij}\),} a proper subset of the former region. This property, coupled with the summation of such terms over \(|E|\) edges, partially explains why the MILP in \eqref{eq:milp_cheeger} runs significantly faster on dense graphs. (see \ref{subsec:milp_cheeger_perf})
\end{remark}

\subsection{Implementation of Cheeger cuts for solving MISDP in \eqref{eq:F_2}}
\begin{theorem}
\label{theorem1} (Cheeger's Inequality \cite{cheeger1970lower}). Let $G = (V, E, \mathbf{w})$ be a weighted undirected graph and let $\lambda_2$ be the {\ac} of graph $G$. Then 
\begin{align*}
    c_f \cdot \lambda_2(G) \leqslant \phi(G) \leqslant \sqrt{2\lambda_2(G)}, \quad \text{where} \ c_f = 1/2.
\end{align*}
\end{theorem}
% \vspace{1em}

Cheeger's inequality is a fundamental result in spectral graph theory that relates the {\ac} $\lambda_2(G)$ of a graph  $G$ to its Cheeger constant $\phi(G)$ (from section \ref{subsec:milp_cheeger}).  Qualitatively, this implies that $\lambda_2(G)$ is large if and only if $\phi(G)$ is large. Leveraging this theorem, we propose the novel \textit{Cheeger cut} based on the left side of the inequality that is valid for the optimal solution of the MISDP in \eqref{eq:F_2}.
% \vspace{1em}

\begin{lemma}
    \label{lem:1} (\textbf{Cheeger Cut}) \textit{Let $\widehat{G}$ be the best-known feasible solution for the MISDP in \eqref{eq:F_2}, with $\lambda_2(\widehat{G})$ representing its algebraic connectivity, serving as a lower bound to $\gamma^*$. Let $c_f > 0$ be a constant, referred to as the Cheeger factor. For any other feasible graph $\tilde{G}$ in the MISDP, let $\tilde{S}$ and $\tilde{V} \setminus \tilde{S}$ denote the partitions of the graph corresponding to $\phi(\tilde{G})$. If $\phi(\tilde{G}) < c_f \cdot \lambda_2(\widehat{G})$, then the following cut is valid for the MISDP's optimal solution $(G^*, \gamma^*)$}:
    $$\sum_{\{i,j\} \in \delta(\tilde{S})} w_{ij} x_{ij} \geqslant c_f \cdot \lambda_2(\widehat{G}) \cdot |\tilde{S}|$$
\end{lemma}
% \vspace{1em}
\begin{proof}
        By Cheeger's inequality, we know that $\phi(G) \geqslant c_f \cdot \lambda_2(G)$ for any feasible graph $G$, where $c_f = 0.5$. In particular, for the optimal graph $G^*$, it holds that $\phi(G^*) \geqslant c_f \cdot \lambda_2(G^*)$. Let $\widehat{G}$ be a sub-optimal, best-known feasible graph for the MISDP. Since $\lambda_2(\widehat{G}) \leq \lambda_2(G^*)$, we also have $\phi(G^*) \geq c_f \cdot \lambda_2(\widehat{G})$. Now, for any sub-optimal graph $\tilde{G}$ obtained within the OA algorithm \ref{alg:cheeger}, if $\phi(\tilde{G}) < c_f \cdot \lambda_2(\widehat{G})$, then the following inequality serves as a valid cut for the optimal graph $G^*$: $\phi(\tilde{G}) \geqslant c_f \cdot \lambda_2(\widehat{G})$, where $\phi(\tilde{G}) = \frac{\sum_{\{i,j\} \in \delta(\tilde{S})} w_{ij} x_{ij}}{|\tilde{S}|}$ from \eqref{eq:cheeger_ratio}. This cut ensures that any sub-optimal $\tilde{G}$ is separated based on its Cheeger constant relative to $\lambda_2(\widehat{G})$, improving the convergence toward $G^*$, thus proving the lemma's claim.
\end{proof}

% \vspace{0.5em}
While a Cheeger factor ($c_f$) of 0.5 is applicable based on Cheeger's inequality, the Cheeger cut in Lemma \eqref{lem:1} remains valid for any $c_f > 0.5$ as long as it satisfies the inequality for the family of weighted graphs under consideration. Empirically, we observe that $c_f > 0.5$ maintains the validity of Cheeger cuts and improves convergence efficiency of Algorithm \ref{alg:cheeger} (see section \ref{subsec:alg_perf}). Establishing theoretically tighter constant-factor bounds on $c_f$ is an open problem and a topic of independent research interest.

Algorithm \ref{alg:cheeger} iteratively refines the outer-approximation (OA) of the MISDP by incorporating both Cheeger and eigenvector cuts \cite{nagarajan2015maximizing, somisetty2024optimal,bhela2021efficient}. The OA procedure solves a sequence of  MILPs, representing polyhedral relaxations of the PSD constraint \eqref{eq:W_psd} in \(\mathcal{F}_1\), thus, significantly improving computational performance. Algorithm \ref{alg:cheeger} solves the MISDP by progressively narrowing the gap between the upper and lower bounds to reach an optimal solution.

\begin{algorithm}
\caption{Proposed Outer-Approximation Algorithm}
\label{alg:cheeger} 
{
\begin{algorithmic}[1]

\State \textbf{Input}: Graph $(V,E,\mathbf{w})$, Cheeger factor $c_f$, Relative optimality tolerance $\varepsilon_{opt} > 0$.

\vspace*{0.05in}
\State \textbf{Initialization}: $LB = 0$, $UB = \infty$, $G^* = \text{null}$, \textcolor{black}{${\lambda_2}(\widehat{G}) =$ Algebraic connectivity of the best-known solution of $\mathcal{F}_1$}, $\mathcal{F}^{u} =$ MILP by dropping the PSD constraint \eqref{eq:W_psd} in $\mathcal{F}_1$. 

\vspace*{0.05in}
\State Solve $\mathcal{F}^{u}$. Let its optimal solution be \textcolor{black}{$({\gamma^u}, \tilde{\mathbf{W}}, \tilde{G})$.}

\vspace*{0.05in}
\State \textcolor{black}{$G^* \longleftarrow \tilde{G}$, \quad $LB \longleftarrow \lambda_2(\tilde{G})$, \quad $UB \longleftarrow {\gamma^u}$}

\vspace*{0.05in}
\State \textcolor{black}{For $\tilde{G}$, evaluate the Cheeger constant $\phi(\tilde{G})$ and the corresponding cut-set $\delta(\tilde{S})$ using the MILP in \eqref{eq:milp_cheeger}.}

\vspace*{0.05in}
\While{$\frac{UB-LB}{UB + 10^{-6}} > \varepsilon_{opt}$}

\vspace*{0.05in}
\State Update $\mathcal{F}^{u}$ with the following linear constraints: 
$$\mathbf{v^- \cdot {W} v^-} \geqslant 0,$$
where $\mathbf{v^-}$ is the violated eigenvector corresponding to the smallest (negative) eigenvalue of $\tilde{\mathbf{W}}$.

\vspace*{0.05in}
\If{$\textcolor{black}{\phi(\tilde{G}) < c_f \cdot \lambda_2(\widehat{G})}$}

\vspace*{0.05in}
\State $\textcolor{black}{\sum_{\{i,j\} \in \delta(\tilde{S})} w_{ij} x_{ij} \geqslant c_f \cdot \lambda_2(\widehat{G}) \cdot |
\tilde{S}|,}$

\EndIf

\vspace*{0.05in}
\State Solve $\mathcal{F}^{u}$. Update $LB$, $UB$ and ${G^*}$.

\EndWhile

\vspace*{0.05in}
\State Return ${G}^*$, lower bound $LB$ and upper bound $UB$.

\end{algorithmic}
}
\end{algorithm}

\section{GREEDY $k$-OPT HEURISTIC}
\label{sec:heuristic}
We now present a greedy heuristic for the {\ac} maximization problem, structured into two stages:
\begin{itemize}
    \item \textbf{Initial feasible graph construction:} In this phase, an initial graph is constructed to meet the required number of edges,  providing a preliminary feasible solution for the MISDP. A pseudo-code of this procedure is given in Algorithm \ref{alg:spanning}.
    \item \textbf{Graph refinement through $k$-opt procedure:} After constructing the initial graph, this phase iteratively enhances the graph's algebraic connectivity using a $k$-opt heuristic. In this process, $k$ edges are added to the graph and an equivalent number are removed to create an improved solution with higher {\ac}. Pseudo-code for this procedure is given in Algorithm \ref{alg:kopt}.
\end{itemize} 

\subsection{Initial feasible graph construction}

The strategy for constructing an initial feasible graph depends on the specific {\ac} maximization problem variant. In the case of sparsifying pose-graphs within SLAM (as in \eqref{eq:F_1}), the approach involves augmenting the top \(q\) edges to a pre-existing connected base graph. These edges are selected based on their ranking, determined by the product \(w_{ij} \cdot (v_i - v_j)^2\), where \(v_i\) is the \(i^{th}\) component of the Fiedler vector of the base graph.

On the other hand, for the variant focusing on spanning tree construction (as in \eqref{eq:F_2}), the methodology is guided by empirical insights from \cite{somisetty2024optimal}, which suggest that spanning trees with higher {\ac} typically have a star-like topology. This topology features a central node with the highest degree, with all other nodes no more than two hops away. The central node is selected by sorting nodes in descending order based on the cumulative sum of the weights of all possible edges that can be augmented to each node. The algorithm then incrementally incorporates nodes into the graph, ensuring they remain within two hops, thereby constructing a  tree with higher {\ac}.

\begin{algorithm}
\caption{Greedy $k$-opt heuristic - Initial feasible graph}
\label{alg:spanning}
{
\begin{algorithmic}[1]

\State \textbf{Input}: Graph $(V,E,\mathbf{w})$, $q$.

\vspace*{0.05in}
\If{base graph exists} 

\vspace*{0.05in}
\State $O_e \longleftarrow$ Sort edges $\{i,j\} \in E^a$ based on the  decreasing order of $w_{ij}(v_i - v_j)^2$ values.

\vspace*{0.05in}
\State Augment top $q$ edges in $O_e$ to the base graph.

\vspace*{0.05in}
\State $\mathcal{T}_0 \longleftarrow$ Initial feasible graph satisfying constraints

\Else 

\vspace*{0.05in}
\State $\mathbf{S}[i] \longleftarrow \text{sum}(\mathbf{w}[i,:] )$

\vspace*{0.05in}
\State $O_n \longleftarrow$ Sort nodes in the decreasing order of $\mathbf{S}[i]$

\vspace*{0.05in}
\State $O_e \longleftarrow$ Sort edges $\{i,j\} \in E$ based on the  decreasing order of $w_{ij}$ values

\vspace*{0.05in}
\State Central node  $\longleftarrow O_n[1]$

\vspace*{0.05in}
\State Add the first edge from the order \(O_e\) to the graph, ensuring it includes the central node.

\vspace*{0.05in}
\State $\mathcal{T}_0 \longleftarrow$ Add edges to the graph in the order specified by $O_e$, ensuring all nodes remain within two hops of the central node until a spanning tree is formed.

\EndIf

\vspace*{0.05in}
\State Return $\mathcal{T}_0$.

\end{algorithmic}
}
\end{algorithm}

\subsection{Graph refinement through $k$-opt procedure}

After constructing an initial feasible graph, the next step is to refine it by generating a new graph within the local neighborhood of the current feasible solution. The goal is to find a new feasible solution with an increased {\ac}. This iterative process continues until no further improvements in {\ac} can be made. Consider \(\mathcal{T}_1\) and \(\mathcal{T}_2\) as two feasible graphs for the {\ac} maximization problem. We define \(\mathcal{T}_2\) as being within the \(k\)-exchange neighborhood of \(\mathcal{T}_1\) if \(\mathcal{T}_2\) can be derived by substituting \(k\) edges from \(\mathcal{T}_1\). During a \(k\)-opt exchange, the objective is to find a new feasible solution within this \(k\)-exchange neighborhood and replace the current solution if the new one has a higher {\ac}.

The \(k\)-opt exchange process involves two key steps: first, selecting a subset of \(k\) edges from the set of augmentable edges to add to the current graph, and then identifying a different subset of \(k\) edges from the updated graph for removal. This approach ensures that the resulting graph remains feasible while aiming to enhance {\ac}. Given the numerous possible combinations for selecting \(k\) edges for both addition and removal, the efficiency of this process is paramount. To address this, we introduce an effective \textit{edge ranking} strategy to create a fast algorithm that efficiently navigates the solution space.

Edges from the augmentable set are prioritized based on the product \(w_{ij} \cdot (v_i - v_j)^2\), where \(v_i\) is the \(i^{th}\) component of the Fiedler vector of the current graph. To streamline selection and improve efficiency, only the top \(m (\geqslant k)\) edges from this list are considered for inclusion in the \(k\)-opt exchange, yielding a subset of \(k\)-edge combinations outlined in step 3 of Algorithm \ref{alg:kopt}. After adding \(k\) edges, the edges in the new graph are reordered according to \(w_{ij} \cdot (v'_i - v'_j)^2\), where \(v'\) represents the Fiedler vector of the updated graph. For removing \(k\) edges, selections are made from the bottom \(m\) of this ranked list (as detailed in step 6 of Algorithm \ref{alg:kopt}), aiming for the most significant increase in {\ac}. The parameter \(m\), adjustable based on computational needs, helps develop a faster greedy algorithm by focusing on a manageable smaller subset of edge combinations.

\begin{algorithm}
\caption{Greedy $k$-opt heuristic - Graph refinement}
\label{alg:kopt}
{
\begin{algorithmic}[1]

\State \textbf{Input}: Graph $(V,E,\mathbf{w})$, $\mathcal{T}_0, k, m$.

\vspace*{0.05in}
\State $\lambda_2^{h} \longleftarrow \lambda_2(\mathbf{L}(\mathcal{T}_0))$

\vspace*{0.05in}
\State $E_{add} \longleftarrow$ Subset of $k$-edge combinations (i.e., $\binom{m}{k}$) %considered for possible addition as obtained by the edge ranking procedure.

\vspace*{0.05in}
\For{each edge combination in $E_{add}$}

\vspace*{0.05in}
\State Add the $k$ edges to the $\mathcal{T}_0$.

\vspace*{0.05in}
\State $E_{del} \longleftarrow$ Subset of $k$-edge combinations (i.e., $\binom{m}{k}$) %considered for possible deletion as obtained by the edge ranking procedure.

\vspace*{0.05in}
\State $\mathcal{T}_1 \longleftarrow$  Feasible graph obtained from deleting the $k$ edges from an edge combination of $E_{del}$

\vspace*{0.05in}
\If{$\lambda_2(\mathbf{L}(\mathcal{T}_0)) > \lambda_2^{h}$}

\vspace*{0.05in}
\State $\mathcal{T}_0 \longleftarrow \mathcal{T}_1$

\vspace*{0.05in}
\State $\lambda_2^{h} \longleftarrow \lambda_2(\mathbf{L}(\mathcal{T}_1))$

\EndIf

\EndFor

\vspace*{0.05in}
\State Return $\mathcal{T}_0$ and $\lambda_2^{h}$.

\end{algorithmic}
}
\end{algorithm}

\section{COMPUTATIONAL RESULTS}
\label{sec:comp}
This section highlights the computational performances of the proposed algorithms via two detailed studies. The first study demonstrates the efficiency of the introduced Cheeger cuts in solving the algebraic connectivity maximization problem. The second study offers a comparative evaluation of the greedy $k$-opt heuristic's performance.

\subsection{Computational settings and test cases}
The proposed optimization formulations and algorithms were modeled and implemented using JuMP v1.2.0 \cite{dunning2017jump} in Julia v1.10.0. Computational analyses utilized Gurobi 11.0.1 \cite{gurobi} as the MILP solver, running on an Apple M2 Max processor with 32GB of memory (personal laptop). User-defined eigenvector and Cheeger cuts were integrated within the branch-and-cut framework via a lazy-cut callback mechanism. All datasets, including the \texttt{CSAIL} and \texttt{Intel} datasets, along with the implemented algorithms are available in the open-source Julia/JuMP package ``\textsc{LaplacianOpt}"\footnote{\label{note:lopt} \url{https://github.com/harshangrjn/LaplacianOpt.jl}}.

\subsection{MILP \eqref{eq:milp_cheeger} vs. MIQCP \eqref{eq:miqcp_cheeger} for Cheeger constant evaluation}
\label{subsec:milp_cheeger_perf}
The comparison between the MILP in \eqref{eq:milp_cheeger} and MIQCP in \eqref{eq:miqcp_cheeger} formulations, both used for evaluating the Cheeger constant of a graph with randomized weights, across different graph sizes and sparsity levels (Table \ref{tab:run_times}) reveals significant insights. A $k$-sparse graph means $k$\% of edges were removed from the complete graph. As alluded to in Remark \ref{rem:milp_region}, the MILP consistently achieved optimality in all runs. In cases where both formulations reached optimality, \textit{the MILP was, on average, 106.39 times faster than the MIQCP}, ranging from 8.28 times faster for 15-node runs to 215.82 times faster for 20-node runs. The highest \textit{speedup observed was 759.04 times} (20 nodes, 1.0 sparsity), with the MILP completing in 0.39 seconds versus MIQCP's 294.38 seconds. For robot localization datasets, the non-sparsified graph's Cheeger constant evaluation in the MILP took 2.1 seconds for \texttt{CSAIL} and 10.6 seconds for \texttt{Intel}, while the MIQCP timed out on both datasets. This substantial speedup is critical for practical applications, where the cheeger constant of a graph is repeatedly evaluated within the separation oracle of Algorithm \ref{alg:cheeger}'s framework.

\input{Inputs/Tables/milp_miqcp_comparison}

\subsection{Algorithm \ref{alg:cheeger}'s performance in finding optimal trees}
\label{subsec:alg_perf}

As discussed in Section \ref{sec:cheeger}, integrating Cheeger cuts with eigenvector cuts enhances computational efficiency in solving MISDPs ($\mathcal{F}_1$) via OA. Table \ref{tab:cheeger_8_10_12} compares the average run times required to reach optimal spanning tree solutions with Algorithm \ref{alg:cheeger}, with and without Cheeger cuts. Here, \(t_1\) represents run times without Cheeger cuts, while \(t_2\) and \(t_3\) indicate run times with Cheeger cuts for factors of 0.8 and 1.0, respectively. Additionally, \(t_i^c\) specifies the time taken to compute the Cheeger constant for each factor, highlighting the efficiency gains from this approach.

\input{Inputs/Tables/cheeger_8_10_12}

The run times (\(t_1\)) for solving 8- and 10-node networks without Cheeger cuts show a noticeable increase in computational complexity as the network size grows. Run times for 12-node networks without Cheeger cuts remain unreported, highlighting the challenges of larger networks and the need for more efficient optimization strategies. Incorporating Cheeger cuts significantly reduces run times, especially for 10- and 12-node instances, with higher Cheeger factors (0.8 to 1.0) further reducing times. This improvement demonstrates the efficacy of Cheeger cuts in enhancing the computational efficiency of Algorithm \ref{alg:cheeger} for larger networks, particularly in cooperative localization. Nevertheless, achieving optimal solutions for the MISDP using both eigenvector and Cheeger cuts remains intractable for larger networks, as encountered in the pose-graph sparsification problem.

% \vspace{-1cm}
\subsection{Performance of the $k$-opt heuristic}
In this section, we evaluate the performance of the $k$-opt heuristic, as introduced in Section \ref{sec:heuristic}, for the two variants outlined in Section \ref{Subsec:F1}. Fig. \ref{fig:spanningtree_k-opt} illustrates the comparison of average gap percentages and run times for spanning tree solutions obtained using 3-opt, 2-opt, and 1-opt heuristics across different problem sizes. These heuristic solutions serve as lower bounds, aiding the application of Cheeger cuts in solving the MISDP for the formulation \(\mathcal{F}_1\). As the optimal solutions are unknown, the gaps relative to the 3-opt solutions are used as a reference for comparison in Fig. \ref{fig:spanningtree_k-opt}.

\input{Inputs/Figures/spanning_trees_k_opt}

In addition, we present a comparative analysis of the $k$-opt heuristic against the MAC algorithm \cite{doherty2022spectral} across different datasets. This comparison includes results on randomly generated synthetic datasets (Fig. \ref{fig:k-opt_MAC}) and real-world benchmark datasets, as shown in Fig. \ref{fig:CSAIL} and Table \ref{tab:so(d)}.

\input{Inputs/Figures/k-opt_vs_MAC}

\input{Inputs/Figures/CSAIL}
Fig. \ref{fig:k-opt_MAC} compares synthetic instances by graph size, number of loop closure edges, and augmentation budget, evaluating {\ac} and run times for 1-opt and 2-opt heuristics against the MAC algorithm \cite{doherty2022spectral}. On average, the proposed greedy heuristics outperform the MAC algorithm across all network sizes. Notably, the 1-opt heuristic exhibits comparable run times to MAC's, whereas 2-opt, albeit slightly more resource-intensive, achieves better outcomes.

Table \ref{tab:so(d)} compares the SO$(d)$ orbit distances for networks sparsified using the 1-opt heuristic and the MAC algorithm \cite{doherty2022spectral}. SO$(d)$ orbit distance measures variations in rotational states between the Maximum-Likelihood Estimator (MLE) of the sparsified and original problems, calculable via singular value decomposition. A smaller orbit distance indicates a closer approximation of the sparsified graph's MLE, relative to the original, reflecting a higher-quality solution \cite{doherty2022spectral}.

This analytical framework highlights the effectiveness of the $k$-opt heuristic in generating higher-quality sparsified graphs, as visually demonstrated in Fig. \ref{fig:CSAIL}. In the figure, sparsified networks (in red) produced by both the 1-opt heuristic and the MAC algorithm \cite{doherty2022spectral} are overlaid on the original network (in blue) from Fig. \ref{fig:sub1}. A comparison reveals a more pronounced blue presence in Fig. \ref{fig:sub3}, indicating that the MAC algorithm's sparsified network deviates more from the original graph, whereas the $k$-opt heuristic's sparsification maintains closer fidelity to the original network structure.

\input{Inputs/Tables/slam_real_data_table}

\section{CONCLUSIONS}

This paper studies the optimization of edge-weighted networks by maximizing algebraic connectivity, a key factor for improving network robustness and precision in robotic network localization. Given the NP-hard nature of this problem, we introduce novel Cheeger cuts to strengthen the mixed-integer semi-definite program (MISDP) formulation, utilizing Cheeger's inequality, which involves repeated evaluation of the Cheeger constant for given graphs. We propose an MI linear program (MILP) for efficiently computing the Cheeger constant, achieving an average speedup of 106.39 times compared to solving the standard MIQCP variant. The integration of Cheeger cuts with eigenvector cuts significantly enhances computational efficiency, enabling optimal solution evaluation for up to 10-node networks and yielding high-quality solutions with reduced computational times for larger networks. Additionally, we introduce a greedy $k$-opt heuristic for quickly sparsifying graphs in cooperative localization and pose-graph SLAM applications, which also serve as valid lower bounds within the Cheeger cuts framework, enhancing scalability.

\label{sec:conc}

\bibliographystyle{IEEEtran}
\bibliography{References}

\end{document}

%% file: Inputs/Tables/milp_miqcp_comparison.tex
\begin{table}[h!]
    \footnotesize
    \centering
    \begin{tabular}{c|c|ccc|ccc}
    % \hline
    \toprule
    \textbf{$n$} & \textbf{Graph} & \multicolumn{3}{c|}{\textbf{MILP - \eqref{eq:milp_cheeger} (sec)}} & \multicolumn{3}{c}{\textbf{MIQCP -- \eqref{eq:miqcp_cheeger} (sec)}} \\ \cline{3-8} 
     & \textbf{sparsity} & \textbf{Min.} & \textbf{Mean} & \textbf{Max.} & \textbf{Min.} & \textbf{Mean} & \textbf{Max.} \\ 
    \hline
    % \midrule
    \multirow{4}{*}{{15}} 
     & 0.4 & 0.01 & 0.05 & 0.08 & 0.03 & 0.09 & 0.16 \\ 
     & 0.6 & 0.04 & 0.13 & 0.18 & 0.09 & 0.60 & 1.26 \\ 
     & 0.8 & 0.09 & 0.13 & 0.18 & 0.59 & 1.52 & 3.82 \\ 
     & 1.0 & 0.12 & 0.15 & 0.20 & 0.89 & 2.22 & 5.48 \\ 
    \hline
    \multirow{4}{*}{{20}} 
     & 0.4 & 0.05 & 0.14 & 0.27 & 0.64 & 1.17 & 1.74 \\  
     & 0.6 & 0.10 & 0.22 & 0.37 & 1.45 & 22.21 & 84.87 \\ 
     & 0.8 & 0.23 & 0.42 & 0.64 & 16.12 & 115.77 & 215.39 \\ 
     & 1.0 & 0.28 & 0.56 & 0.86 & 14.96 & 268.49 & 456.20 \\ 
     \hline
    \multirow{4}{*}{{25}} 
     & 0.4 & 0.02 & 0.25 & 0.53 & 0.88 & 16.74 & 47.00 \\
     & 0.6 & 0.40 & 0.97 & 1.65 & \multicolumn{3}{c}{Time Limit} \\  
     & 0.8 & 1.10 & 2.00 & 3.39 & \multicolumn{3}{c}{Time Limit} \\ 
     & 1.0 & 1.17 & 2.31 & 4.15 & \multicolumn{3}{c}{Time Limit} \\ 
    %  \hline
    \bottomrule
    \end{tabular}
    \caption{Comparison of MILP and MIQCP run times across different node sizes and sparsity levels where the runs did not time out. Time Limit was set to 500 sec, with statistics based on 50 randomized instances.}
    \label{tab:run_times}
    \vspace{-0.1cm}
\end{table}

%% file: Inputs/Tables/cheeger_8_10_12.tex
\begin{table}[ht!]
\centering
%\small
\caption{Comparison of average run times: Algorithm \ref{alg:cheeger}  with versus without Cheeger cuts to achieve optimal spanning trees across instances with 8, 10, and 12 nodes. Time Limit was set to 3600 sec.}
\label{tab:cheeger_8_10_12}
\resizebox{\columnwidth}{!}{%
\begin{tabular}{c c c c c c} 
\toprule
\multirow{2}{*}{$n$}  & {w/o Cheeger cuts}  & \multicolumn{2}{c}{$c_f = \frac{0.8 \cdot \phi(\widehat{G})}{\lambda_2(\widehat{G})}$} & \multicolumn{2}{c}{$c_f = \frac{\phi(\widehat{G})}{\lambda_2(\widehat{G})}$} \\ 
\cmidrule(lr){3-4} \cmidrule(lr){5-6}
 & $t_1$ (sec) & $t_2$ (sec) & $t_2^c$ (sec) & $t_3$ (sec) & $t_3^c$ (sec) \\
\midrule
8  & 1.4    &   0.6   &  0.2   &  0.3    &  0.1  \\[1ex]
10 & 502.3  &   24.1  &  1.0   &   4.5   & 0.6  \\[1ex]
12 & Time Limit     & 1752.2 & 15.3 & 160.5 & 2.5 \\
\bottomrule
\end{tabular}%
}
\end{table}

%% file: Inputs/Figures/spanning_trees_k_opt.tex
%\begin{table}[ht!]
%\centering
%\caption{Comparative analysis of heuristics (3-opt, 2-opt, and 1-opt) with $m=20$ for spanning tree solutions across different network sizes ($n$). The gap and run times are averages from fifty random instances for each $n$.}
%\label{tab:spanningtree_k-opt}
%\resizebox{\columnwidth}{!}{
%\begin{tabular}{@{}ccccccc@{}} 
%\toprule
%\textbf{Nodes ($n$)} & \multicolumn{2}{c}{\textbf{3-opt}} & \multicolumn{2}{c}{\textbf{2-opt}} & \multicolumn{2}{c}{\textbf{1-opt}} \\ 
%\cmidrule(lr){2-3} \cmidrule(lr){4-5} \cmidrule(lr){6-7}
% & \textbf{Gap (\%)} & \textbf{Run time (s)} & \textbf{Gap (\%)} & \textbf{Run time (s)} & \textbf{Gap (\%)} & \textbf{Run time (s)} \\
%\midrule
%15 & 0 & 14.2 & 0 & 4.8 & 2.9 & 0.1 \\[1ex]
%25 & 0 & 27.0 & 4.6 & 18.2 & 9.4 &  1.0 \\[1ex]
%40 & 0 & 43.8 & 8.9 & 24.6 & 12.9 & 7.4 \\[1ex]
%60 & 0 & 153.7 & 0.7 & 89.8 & 7.4 & 28.9 \\
%\bottomrule
%\end{tabular}
%}
%\end{table}

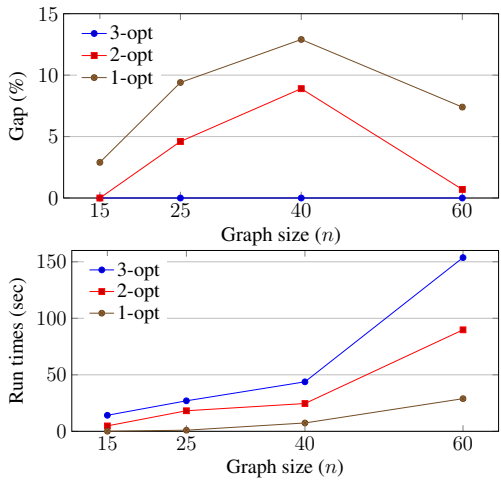
\begin{figure}
\centering
\resizebox{0.75\columnwidth}{!}{
\begin{tikzpicture}
\begin{axis}[
    xlabel={Graph size ($n$)},
    xlabel style={font=\Large},
    ylabel={Gap (\%)},
    ylabel style={font=\Large},
    title style={font=\Large},
    tick label style={font=\Large},
    xtick={15, 25, 40, 60},
    legend pos=north west,
    legend style={draw=none, font=\Large},
    ymajorgrids,
    width=12cm,
    height=6cm,
    ymin=0,
    ymax=15
]

\addplot table [x=Nodes, y=Gap_3opt, col sep=comma] {data_kopt.csv};
\addplot table [x=Nodes, y=Gap_2opt, col sep=comma] {data_kopt.csv};
\addplot table [x=Nodes, y=Gap_1opt, col sep=comma] {data_kopt.csv};

\legend{3-opt, 2-opt, 1-opt}

\end{axis}
\end{tikzpicture}}
\resizebox{0.75\columnwidth}{!}{
\begin{tikzpicture}
\begin{axis}[
    xlabel={Graph size ($n$)},
    xlabel style={font=\Large},
    ylabel style={font=\Large},
    ylabel={Run times (sec)},
    title style={font=\Large},
    tick label style={font=\Large},
    xtick={15, 25, 40, 60},
    legend pos=north west,
    legend style={draw=none,font=\Large},
    ymajorgrids,
    width=12cm,
    height=6cm,
    ymin=0,
    ymax=160
]

\addplot table [x=Nodes, y=Runtime_3opt, col sep=comma] {data_kopt.csv};
\addplot table [x=Nodes, y=Runtime_2opt, col sep=comma] {data_kopt.csv};
\addplot table [x=Nodes, y=Runtime_1opt, col sep=comma] {data_kopt.csv};

\legend{3-opt, 2-opt, 1-opt}

\end{axis}
\end{tikzpicture}%
}
\caption{Comparative analysis of heuristics (3-opt, 2-opt, and 1-opt) with $m=20$ for spanning tree solutions across different network sizes ($n$). The gap and run times are averaged over fifty random instances for each $n$.}
\label{fig:spanningtree_k-opt}

\end{figure}

%% file: Inputs/Figures/k-opt_vs_MAC.tex
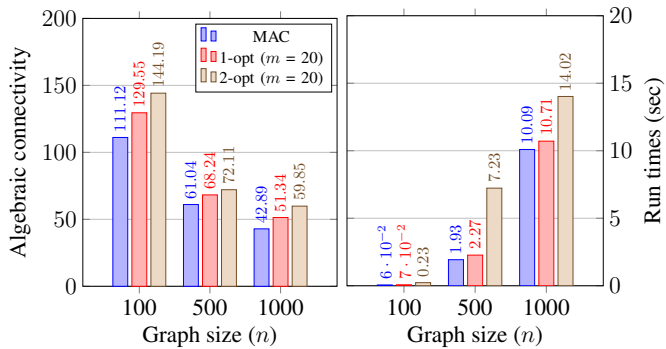
\begin{figure}
    \centering
    \resizebox{0.5\columnwidth}{!}{%
\begin{tikzpicture}
\begin{axis}[
    x tick label style={/pgf/number format/1000 sep=},
    ybar=2.5pt,
    bar width=9pt,
    x=1.5cm,
    ymin=0,
    %axis on top,
    ymax=200,
    xtick=data,
    enlarge x limits=0.4,
    ylabel={Algebraic connectivity},
    xlabel={Graph size ($n$)},
    xticklabels={{100}, {500}, {1000}},
    ymajorgrids,
    nodes near coords,
    nodes near coords align={vertical},
    every node near coord/.append style={rotate=90, anchor=west},
    label style={font=\Large},
    tick label style={font=\Large}  
    ]
\addplot coordinates {( 1,111.12	) ( 2,61.04	)  (3, 42.89)};
\addplot  coordinates {( 1,129.55) ( 2,	68.24	) ( 3,51.34	) };
\addplot  coordinates {( 1,144.19) ( 2,	72.11	) ( 3,59.85	) };
\legend{MAC, 1-opt ($m$ = 20), 2-opt ($m$ = 20)}
\end{axis}
\end{tikzpicture}%
}
\resizebox{0.485\columnwidth}{!}{%
\begin{tikzpicture}
\begin{axis}[
    x tick label style={/pgf/number format/1000 sep=},
    ybar=2.5pt,
    bar width=9pt,
    yticklabel pos=right,
    x=1.5cm,
    ymin=0,
    %axis on top,
    ymax=20,
    xtick=data,
    enlarge x limits=0.4,
    ylabel={Run times (sec)},
    xlabel={Graph size ($n$)},
    xticklabels={{100}, {500}, {1000}},
    ymajorgrids,
    nodes near coords,
    nodes near coords align={vertical}, % This aligns the nodes vertically
    every node near coord/.append style={rotate=90, anchor=west}, % Rotate and adjust the position of the value labels for better alignment and visibility
    label style={font=\Large},
    tick label style={font=\Large}  
    ]
    \addplot coordinates {( 1,0.06	) ( 2,1.93	)  (3, 10.09)};
\addplot  coordinates {( 1,0.07) ( 2,	2.27	) ( 3,10.71	) };
\addplot  coordinates {( 1,0.23) ( 2,	7.23	) ( 3,14.02	) };
%\legend{1-opt ($m$ = 10), MAC}
\end{axis}
\end{tikzpicture}%
}
\caption{Comparison of the greedy 1-opt and 2-opt heuristics with parameter $m = 20$ against the MAC algorithm \cite{doherty2022spectral} for different graph sizes. The left graph displays the algebraic connectivity, while the right graph illustrates the run times.}
\label{fig:k-opt_MAC}
\vspace{-0.3cm}
\end{figure}

%% file: Inputs/Figures/CSAIL.tex
\begin{figure*}[ht]
  \centering
  \begin{subfigure}{.24\linewidth}
    \centering
    \includegraphics[width=1.0\linewidth]{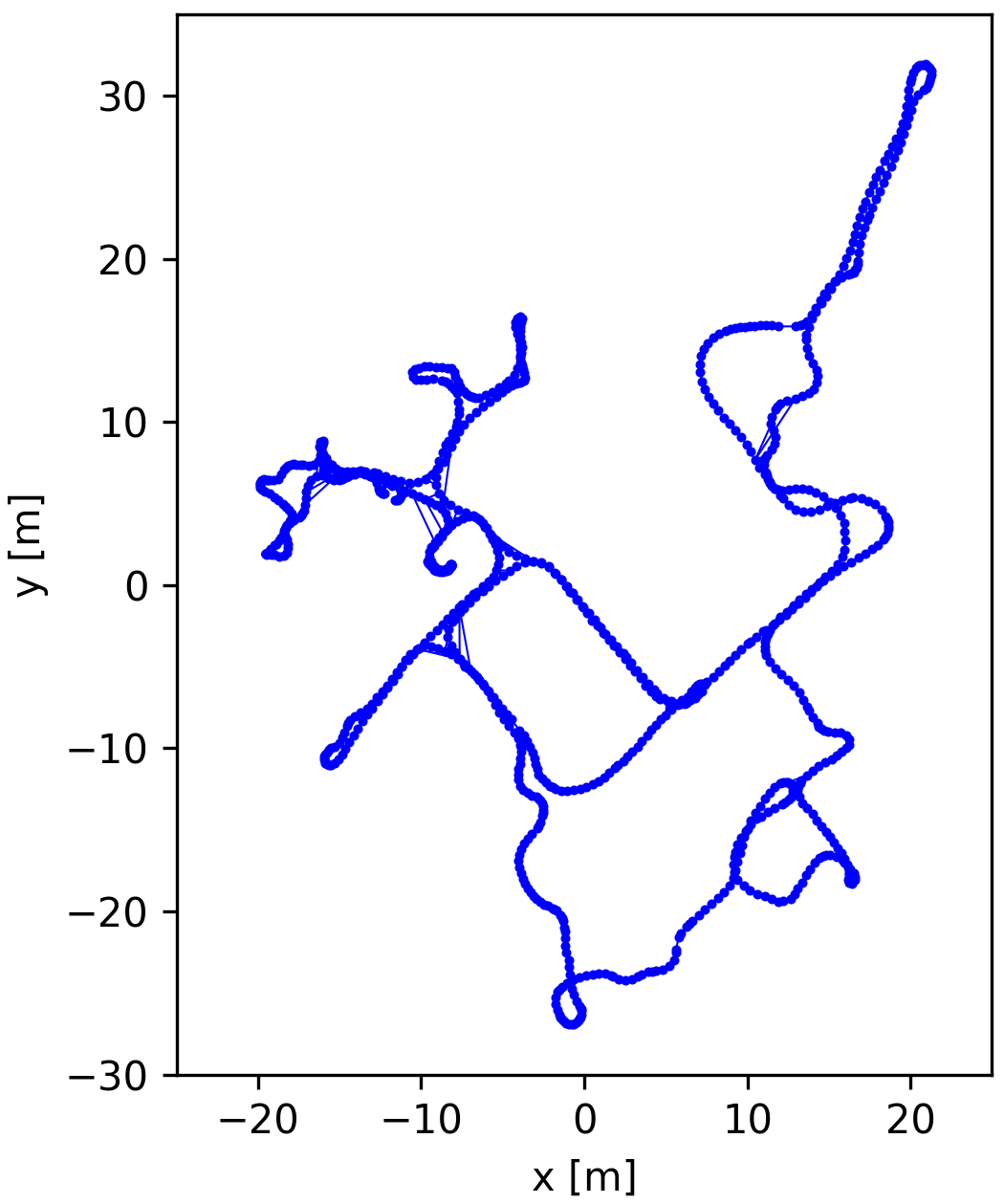}
    \caption{Original graph}
    \label{fig:sub1}
  \end{subfigure}%
  \hfill % Optional: this adds spacing between the subfigures
  \begin{subfigure}{.24\linewidth}
    \centering
    \includegraphics[width=1.0\linewidth]{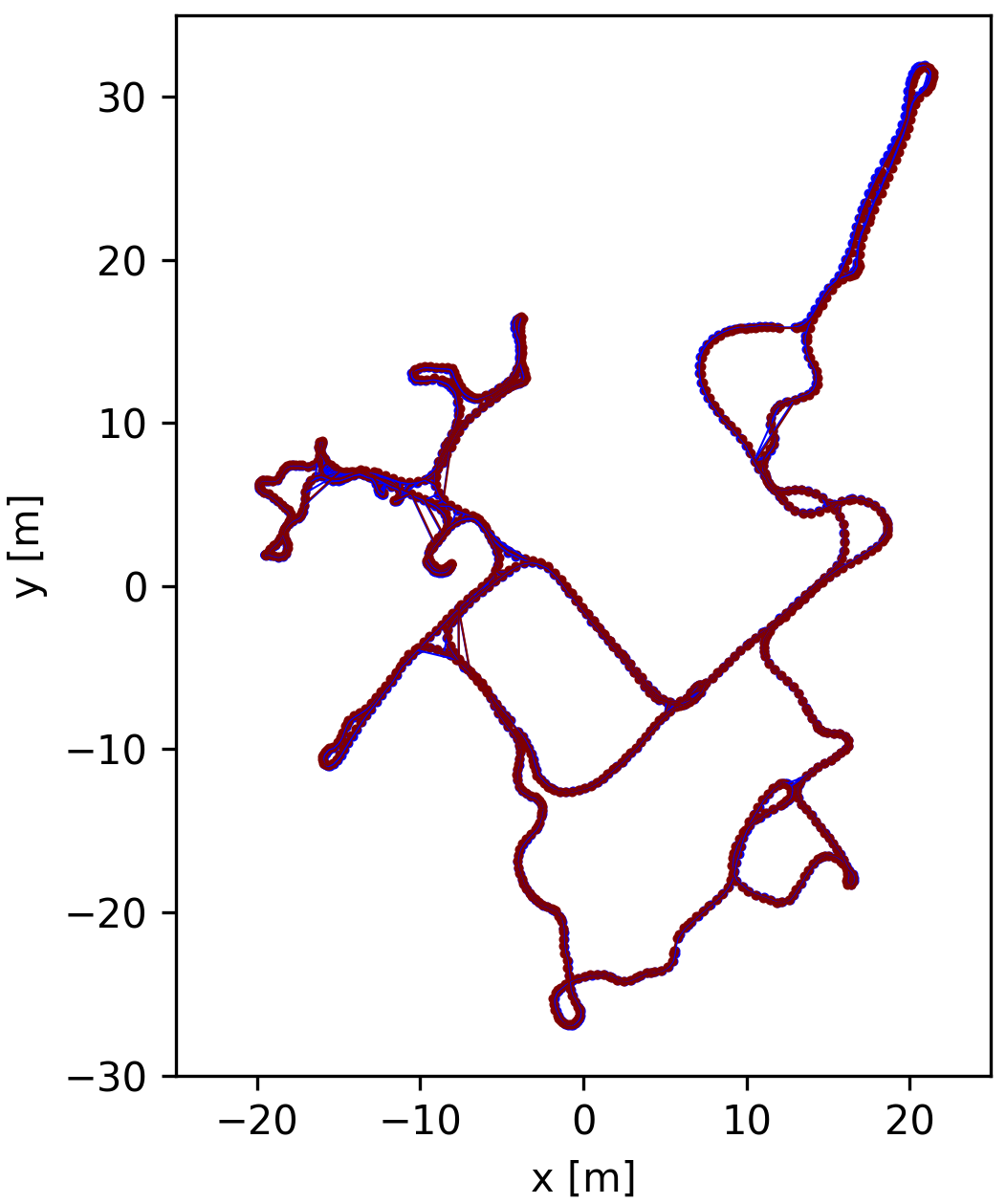}
    \caption{1-opt heuristic algorithm}
    \label{fig:sub2}
  \end{subfigure}%
  \hfill % Optional: this adds spacing between the subfigures
  \begin{subfigure}{.24\linewidth}
    \centering
    \includegraphics[width=1.0\linewidth]{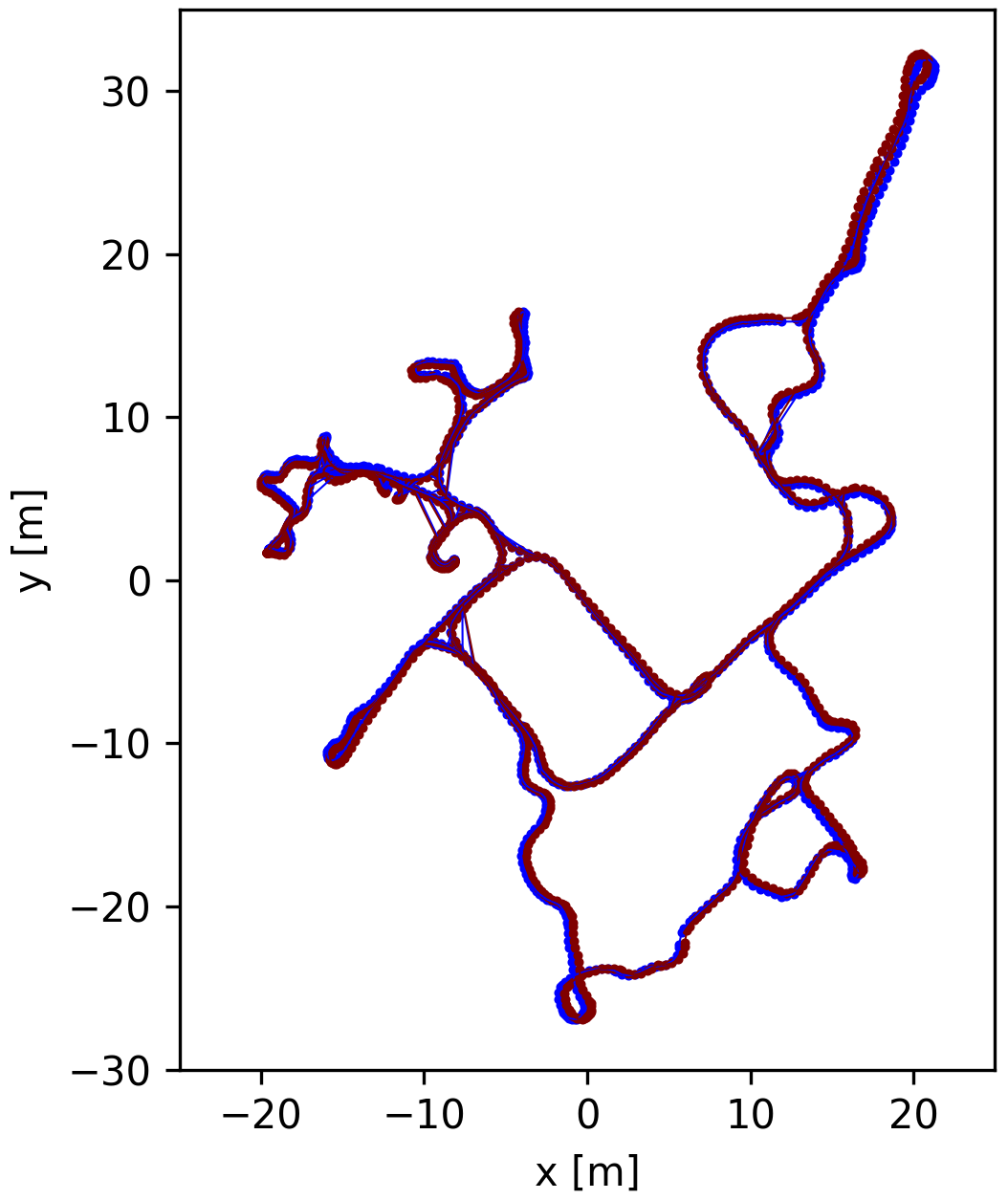}
    \caption{MAC algorithm}
    \label{fig:sub3}
  \end{subfigure}
  \begin{subfigure}{.24\linewidth}
    \centering
    \includegraphics[width=1.0\linewidth]{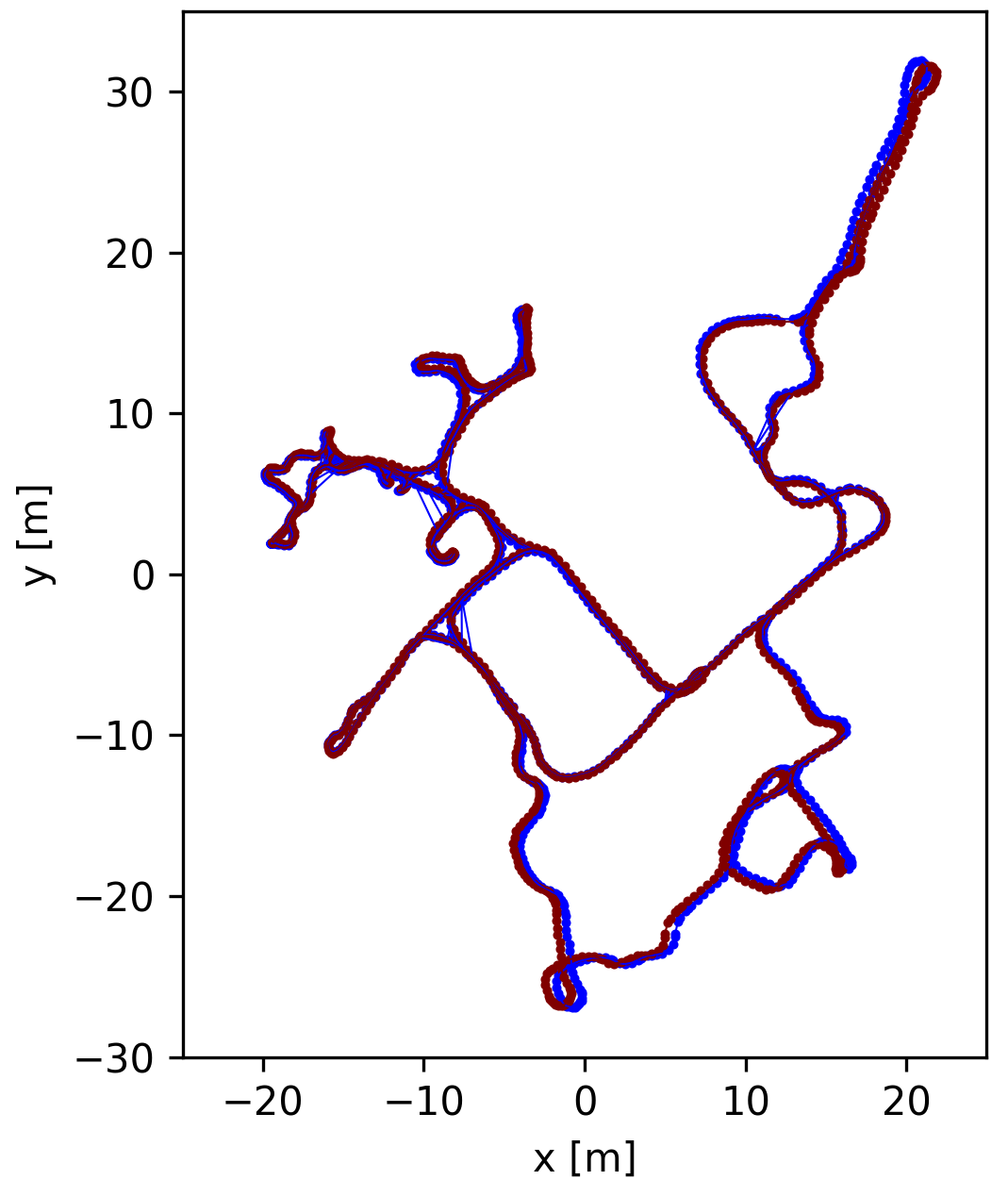}
    \caption{Random sparsification}
    \label{fig:sub4}
  \end{subfigure}
  \caption{Pose-graph optimization for the \texttt{CSAIL} dataset. The sequence shows: (a) the original graph with all measurements; (b) 20\% edge sparsification using greedy 1-opt heuristic to maximize {\ac} with $m=30$; (c) 20\% edge sparsification using the MAC algorithm \cite{doherty2022spectral}; and (d) 20\% edge sparsification through random edge selection. In (b), (c), and (d), the sparsified networks (red) are superimposed on the original graph (blue), facilitating a visual comparison of the sparsification effects.}
  \label{fig:CSAIL}
\end{figure*}

%% file: Inputs/Tables/slam_real_data_table.tex
\begin{table}[ht!]
\centering
\caption{SO$(d)$ orbit distance comparison for benchmark datasets between the 1-opt heuristic with $m=30$ and the MAC algorithm \cite{doherty2022spectral} under various sparsification levels.}
\resizebox{0.9\columnwidth}{!}{%
    \begin{tabular}{l c c c c} % Using 'l' for left alignment of the Dataset name
    \toprule
    \textbf{Dataset} & \textbf{Edges added} (\%) & \textbf{1-opt }($m=30$) & \textbf{ MAC}  & \textbf{Random}\\
    \midrule
    \texttt{Intel} & 25 & 0.193 & 0.221 & 0.435 \\
    $n = 1728$ & 50 & 0.164 & 0.171 & 0.318 \\
    $|E^a| = 400$ & 75 & 0.059 & 0.063 & 0.155\\
    \midrule
    \texttt{CSAIL} & 25 & 0.161 & 0.162 & 0.581\\
    $n = 1045$ & 50 & 0.099 & 0.102 & 0.316\\
    $|E^a| = 128$ & 75 & 0.040 & 0.041 & 0.044\\
    \bottomrule
    \end{tabular}%
}
\label{tab:so(d)}
\end{table}
\vspace{-0.5em}